 \documentclass{llncs}
\newcommand{\longversion}[1]{#1}
\newcommand{\shortversion}[1]{}

\usepackage{hyperref}

\usepackage{subfigure}
 \usepackage{amsmath}
 \usepackage{amssymb}
\usepackage{psfrag}
\usepackage{graphicx}
\usepackage{pifont}
\usepackage{url}
\usepackage{algorithm,myalgorithmic}
\usepackage{boxedminipage}
\usepackage{multirow}
\usepackage{tikz}
\usepackage{xspace}

\renewcommand{\algorithmicrequire}{\textbf{Input(s):}}
\renewcommand{\algorithmicensure}{\textbf{Output(s):}}

\newcommand{\mc}{\mathcal}

\newtheorem{fact}{Fact}

\newcommand{\comment}[1]{ }
\newcommand{\scl}{0.67}

\newcommand{\Oh}{{\mc{O}}}

\renewcommand{\int}{\operatorname{int}}       

\newcommand{\FPT}{\mbox{$\mc{FPT}$}}
\newcommand{\NP}{\mbox{$\mc{NP}$}}

\newcommand{\emphprob}[3]
{\begin{center}\begin{boxedminipage}{0.95 \textwidth}\noindent {{\bf {\sc  #1}} \\
{\bf Given:} #2.\\
{\bf Task:} #3.}\end{boxedminipage}\end{center}\medskip}

\newcommand{\redrule}[2]{\vspace*{1ex}{\bf #1}: #2}
\newcommand{\mist}{\textsc{MIST}\xspace}
\newcommand{\rt}{1.8669}
\newcommand{\rtk}{3.4854}
\newcommand{\rtki}{2.7321}
\newcommand{\rtkii}{2.1364}

\begin{document}

\title{Exact and Parameterized Algorithms for\\ {\sc Max Internal Spanning Tree}%
\thanks{\longversion{This work was partially s}\shortversion{S}upported by a PPP grant between DAAD (Germany) and NFR (Norway).}}
\author{
 Henning Fernau\inst1,
 Serge Gaspers\inst2,
 Daniel Raible\inst1      
}
\institute{%
Univ. 
Trier, FB 4---Abteilung Informatik, 
D-54286 Trier,
Germany\\
\email{\{fernau,raible\}@uni-trier.de}
\and
LIRMM -- Univ. of Montpellier 2, CNRS, 34392 Montpellier, France\\
\email{gaspers@lirmm.fr}
}

\maketitle
\begin{abstract}We consider the \NP-hard problem of finding a spanning tree with a maximum number of internal vertices. This problem is a
generalization of the famous 
 {\sc Hamiltonian Path} problem. Our dynamic-programming algorithms for general and degree-bounded graphs have running times of the form $\Oh^*(c^n)$
($c \le 3$). The main result, however, is a branching algorithm for graphs with maximum degree three. It only needs polynomial space and has a running
time of $\Oh^*(\rt^n)$ when analyzed with respect to the number of vertices. We also show that its running time is $\rtkii^k n^{\Oh(1)}$ when the goal
is to find a spanning tree with at least $k$ internal vertices. Both running time bounds are obtained via a Measure \& Conquer analysis, the latter
one being a novel use of this kind of analyses for parameterized algorithms.
\end{abstract}

\section{Introduction}
\subsubsection*{Motivation}
In this paper we investigate the following problem:
\emphprob{Max Internal Spanning Tree (MIST)}{A graph $G=(V,E)$ with $n$ vertices and $m$ edges}{Find a spanning tree of $G$ with a maximum number of internal vertices}
\mist is a generalization of the \longversion{famous and }well-studied {\sc Hamiltonian Path} problem\shortversion{:}\longversion{. Here, one is asked to} find a path in a graph such that every vertex is visited exactly  once. Clearly, such a path, if it exists, is also a spanning tree, namely one with a maximum number of internal vertices. Whereas the running time barrier of $2^n$ has not been broken for general graphs, there are faster algorithms for cubic graphs (using only polynomial space). It is natural to ask if for the generalization, \mist, this can also be obtained.

A second issue is \longversion{if we can}\shortversion{to} find an algorithm for \mist with a running time of the form $\Oh^*(c^n)$.
\footnote{\longversion{Throughout the paper, we write }$f(n) = \Oh^*(g(n))$ if $f(n) \leq p(n) \cdot g(n)$
for some polynomial $p(n)$.}  
The \longversion{very }na{\"i}ve approach gives only an upper bound of $\Oh^*(2^m)$.

A possible application could be the following scenario. \longversion{Suppose you have a set of}\shortversion{Consider} cities which should be connected with water pipes. The possible connections between them can be represented by a graph $G$. It suffices to compute a spanning tree $T$ for $G$.  In $T$ we may have high degree vertices that have to be implemented by branching pipes\longversion{. These branching pipes}\shortversion{ which} cause turbulences and therefore pressure may drop. To minimize the number of branching pipes one  can equivalently compute a spanning tree with the smallest number of leaves, leading to \mist. 
Vertices representing 
 branching pipes should not be of  arbitrarily high degree, motivating us to investigate \mist on degree-restricted graphs.

\subsubsection*{Previous Work}
It is well-known that the more restricted problem, \textsc{Hamiltonian Path}, can be solved within $\Oh(n^2 2^n)$ steps and exponential space. This
result has been  independently obtained by Bellman~\cite{Bel62}, and Held and Karp~\cite{HelKar62}. The \textsc{Traveling Salesman} problem is very
closely related to \textsc{Hamiltonian Path}. Basically, the same algorithm solves this problem, but there has not been any improvement on the running
time since 1962. The space requirements have, however, been improved and now there 
are $\Oh^*(2^n)$ algorithms needing only polynomial space. In 1977, Kohn \emph{et al.}~\cite{KohGotKoh77} gave an algorithm based on generating
functions with a running time of $\Oh(2^n n^3)$ and space requirements of $\Oh(n^2)$ and in 1982 Karp~\cite{Kar77} came up with an algorithm which
improved storage requirements to $\Oh(n)$ and preserved this run time by an inclusion-exclusion approach.

Eppstein~\cite{Epp03} studied the {\sc Traveling Salesman }problem on cubic graphs. He could achieve a running time of $\Oh^*(1.260^n)$ using
polynomial space. Iwama and Nakashima~\cite{IwaNak07} could improve this to $\Oh^*(1.251^n)$. 
solving {\sc Hamiltonian Path} in $\Oh^*(1.251^n)$.
Bj{\"o}rklund \emph{et al.}~\cite{BjoHusKasKoi08} studied TSP with respect to degree-bounded graphs. Their algorithm is a variant of the classical
$2^n$-algorithm and the space requirements are therefore exponential. Nevertheless, they showed that for a graph with maximum degree 
$d$ there is a $\Oh^*((2-\epsilon_d)^n)$-algorithm. In particular for $d=4$ there is a $\Oh(1.8557^n)$- and for $d=5$ a $\Oh(1.9320^n)$-algorithm.

\mist was also studied with respect to parameterized complexity. The (standard) parameterized version of the problem is parameterized by $k$, and asks
whether $G$ 
has a spanning tree with at least $k$ internal vertices. Prieto and Sloper~\cite{PriSlo2003} proved a $\Oh(k^3)$-vertex kernel for the problem showing
\FPT-membership. In \cite{\longversion{Pri2005,}PriSlo2005} the kernel size has been improved to $\Oh(k^2)$ and in \cite{FominGST} to $3k$.
Parameterized algorithms for \mist have been studied in~\cite{CohenFGKSY09,FominGST,PriSlo2005}. Prieto and Sloper~\cite{PriSlo2005} gave the first
FPT algorithm, with running time $2^{4k \log k} \cdot n^{\Oh(1)}$. This result was improved by Cohen {\em et al.}~\cite{CohenFGKSY09} who solve a more
general directed version of the problem in time $49.4^k \cdot n^{\Oh(1)}$. The current fastest algorithm has running time $8^k \cdot
n^{\Oh(1)}$~\cite{FominGST}.

Salamon~\cite{Sal07} studied the problem considering approximation. He could achieve a $\frac{7}{4}$-approximation. A $2(\Delta-2)$-approximation for the node-weighted version is also a by-product. Cubic and claw-free graphs were considered by Salamon and Wiener~\cite{SalWie08}. They introduced algorithms with approximation ratios $\frac{6}{5}$ and $\frac{3}{2}$, respectively.

\subsubsection{Our Results} This paper gives two algorithms:\shortversion{\\}
\longversion{
\begin{enumerate}\itemsep.5pt
 \item[$(a)$]}
\shortversion{(a)} A dynamic-programming algorithm solving \mist in time $\Oh^*(3^n)$. We extend this algorithm and show that for any degree-bounded
graph a running time of $\Oh^*((3-\epsilon)^n)$ with $\epsilon >0$ can be achieved. To our knowledge this is the first algorithm for \mist with a
running time 
bound of the form $\Oh^*(c^n)$.\footnote{Before the camera-ready version of this paper was prepared, Nederlof~\cite{Nederlof09} came up with a
polynomial-space $\Oh^*(2^n)$ algorithm for \mist on general graphs, answering a question in a preliminary version of this paper.}
\shortversion{\\(b)}
 \longversion{\item[$(b)$]} A  \shortversion{polynomial-space }branching algorithm solving the maximum degree $3$ case in time $\Oh^*(\rt^n)$.
\longversion{The space requirements are only polynomial in this case. }We also analyze the same algorithm from a parameterized point of view,
achieving a running time of $\rtkii^k n^{\Oh(1)}$ to find a spanning tree with at least $k$ internal vertices (if
\shortversion{possible}\longversion{the graph admits such a spanning tree}). The latter analysis is novel in a sense that we use a potential function
analysis---Measure \& Conquer---in a way that, to our knowledge, is much less restrictive than any previous analyses for parameterized algorithms that
were based on the potential function method.
\longversion{\end{enumerate}}

\subsubsection{Notions \& Definitions}
We consider only simple undirected graphs $G=(V,E)$. The \emph{neighborhood} of a vertex $v\in V$ in $G$ is $N_G(v):=\{u \mid \{u,v\} \in E\}$ and its \emph{degree} is $d_G(v):=|N_G(v)|$. The \emph{closed neighborhood} of $v$ is $N_G[v]:=N_G(v) \cup \{v\}$ and for a set $V' \subseteq V$ we let $N_G(V'):=\left(\bigcup_{u \in V'}N_G(u)\right) \setminus V'$.
We omit the subscripts of $N_G(\cdot)$, $d_G(\cdot)$, and $N_G[\cdot]$ when $G$ is clear from the context.
A \emph{subcubic graph} has maximum degree at most three. 
For a (partial) spanning tree $T$ let $I(T)$ be the set of its internal (non-leaf) vertices and $L(T)$ the set of its leaves. An \emph{$i$-vertex} $u$ is a vertex with $d_T(u)=i$ with respect to some spanning tree $T$. The \emph{tree-degree} of some $u \in V(T)$ is $d_T(u)$. $T_o$ refers to an arbitrary maximum internal spanning tree. We also speak of the \emph{$T$-degree} $d_T(v)$ when we refer to a specific spanning tree. A \emph{Hamiltonian path} is a sequence of pairwise distinct vertices $v_1,\ldots ,v_n$ from $V$ such that $\{v_i,v_{i+1}\} \in E$ for $1 \le i \le n-1$.

\section{The Problem on General Graphs}
We give a simple dynamic-programming algorithm to solve MIST within $\Oh^*(3^n)$ steps. Here we build up a table $M[I,L]$ with $I,L \subset V$ such that $I \cap L= \emptyset$. The set $I$ represents the internal vertices and $L$ the leaves of some tree with vertex set $I \cup L$ in $G$. If such a tree exists then we have $M[I,L]=1$ and otherwise a zero-entry. \longversion{In the beginning, we initialize
all table-entries with zeros.} 
In the initializing phase we iterate over all $e \in E$ and set $M[\emptyset,e]=1$. 
Note that every edge is a tree with two leaves and no internal vertices. 
To compute further entries we use  dynamic programming in stages $3,\ldots,n$. 
Stage $i$ consists in determining all table entries indexed by all $I,L \subseteq V$ with $|I|+|L|=i$ and $I \cap L = \emptyset$ such that $G[I \cup L]$ is connected and $M[I,L]=1$. We obtain the table entries of stage~$i$ by inspecting the non-zero entries of stage~$(i-1)$. 
\longversion{If $|I|+|L|=i-1$ and $M[I,L]=1$ then for every $x \in N(I \cup L)$ consider any possibility of attaching $x$ as a leaf to the tree formed by $I \cup L$. There are two possibilities:
\begin{enumerate}
 \item[$a)$] $x$ is adjacent to an internal vertex, then set $M[I,L \cup \{x\}]=1$, and
 \item[$b)$] $x$ is adjacent to a leaf $y$ then set $M[I \cup \{y\},(L \setminus \{y\}) \cup \{x\}]=1$.
\end{enumerate}
}Recursively this can be expressed as follows:
\begin{equation}\label{one}
M[I,L]=\left\{
\begin{array}{l@{\,:\,}l}
1 & \exists x \in L \cap N(I): M[I,L \setminus  \{x\}]=1\\
1 & \exists x \in L, y \in N(x) \cap I: M[ I \setminus \{y\},(L \cup \{y\}) \setminus \{x\}]=1 \\
0 & \text{otherwise}
\end{array}
 \right.
\end{equation}
Here we use the fact that, if we delete a leaf $x$ of a tree $T$, then there are two possibilities for the resulting tree $T'$: Either $T'$ has the same internal vertices as $T$ but one leaf less, or the father $y$ of $x$ in $T$ has become a leaf as $d_T(y)=2$. These are \longversion{exactly }the two cases which are considered in Eq.~(\ref{one}). The number of entries in $M$ is at most $\longversion{\sum_{A,B \subseteq V \atop A \cap B = \emptyset} 1 = \sum_{D \subseteq V} \sum_{C \subseteq D}1=}3^{|V|}$.
\begin{lemma}
{\sc Max Internal Spanning Tree} can be solved in time $\Oh^*(3^n)$.
\end{lemma}

\subsubsection{Bounded Degree}
In this paper, we are particularly interested in solving MIST on graphs of bounded degree. The next lemma is due to \cite{BjoHusKasKoi08}.
\begin{lemma}\label{noconnect}
An $n$-vertex graph with maximum vertex degree $\Delta$ has at most $\beta_\Delta^n+n$ connected vertex sets with $\beta_\Delta=(2^{\Delta+1}-1)^{\frac{1}{\Delta+1}} $.
\end{lemma}
In particular, $n$ refers to the connected sets of size one, which is $\{\{x\}\mid x\in V\}$. 
Thus, the number of all connected sets of size greater than one is $\beta_\Delta^n$.  
Using this 
we 
prove: 
\begin{lemma}\label{3eps}
For any $n$-vertex graph with maximum degree $\Delta$ there is an algorithm that solves MIST in time $\Oh^*(3^{(1-\epsilon_\Delta)n})$ with $\epsilon_\Delta >0$.
\end{lemma}

\begin{proof}
As Lemma~\ref{noconnect} bounds the number of connected subsets of $V$, we would like to skip unconnected ones. This is guaranteed by the approach of dynamic programming in stages.
Let $\cal C$ consist of the sets $F \subseteq V$ such that $G[F]$ is connected and $|F|\ge2$. Then the number of visited entries of $M[I,L]$ with $|I|\ge 2$ in all stages is at most
\begin{equation*}
\sum_{A \subseteq V \atop A \in {\cal C}}\sum_{I \subseteq A \atop I \in {\cal C}} 1\le \sum_{A \subseteq V \atop A \in {\cal C}} \beta_\Delta^{|A|} \le \sum_{i=0}^n \left(n \atop i\right) \beta_\Delta^i =(\beta_\Delta+1)^n
\end{equation*}

\begin{table}[tb]\centering
\begin{tabular}{c c c c c c c}\hline
$\Delta$ & 3 & 4 & 5 & 6 & 7 & 8 \\ \hline
Running Time & 2.9680 & 2.9874 & 2.9948 & 2.9978 & 2.9991 & 2.9996\\ \hline
\end{tabular}
\caption{\label{tab1}Running times for graphs with maximum degree  $\Delta$.}
\end{table}

The visited entries $M[I,L]$ where $|I|=1$ is $n$. As $\beta_\Delta<2$ for any\longversion{ constant}~$\Delta$, this shows Lemma~\ref{3eps}. Table~\ref{tab1} gives an overview on the running times for small\longversion{ values of}~$\Delta$.
\qed
\end{proof}

A na{\"i}ve approach to solve the degree restricted version of MIST is to consider each edge-subset. The running time is $\Oh^*(2^{\frac{\Delta}{2}n})$ where $\Delta$ is the maximum degree.
Compared to Table~\ref{tab1}, \longversion{we see that }for every $\Delta \ge 4$, the na{\"i}ve algorithm is slower. A further slight improvement for
$\Delta=3$ provides the next observation. The line graph $G_l$ of $G$ has maximum degree four and hence there are no more than $\beta_4^{|V(G_l)|}$
connected vertex subsets. Clearly, $G$ then has no more than $\beta_4^{|E(G)|}$ connected edge subsets. Having already a partial connected solution
$T_E \subseteq E$ we only branch on edges $\{u,v\}$ with $u \in T_E$ and $v \not \in T_E$. 
Thus, the run time is $\Oh^*(\beta_4^{\frac{3}{2}n}) = \Oh(2.8017^n)$. We can easily generalize this for arbitrary degree $\Delta$ to
$\Oh^*(\beta_{2\Delta-2}^{\frac{\Delta}{2}n})$.

\section{Subcubic Maximum Internal Spanning Tree}
\subsection{Observations}
Let $t^T_i$ denote the number of vertices $u$ such that $d_T(u)=i$ for a spanning tree $T$. Then the following proposition can be proved
by induction on the number of vertices.
\begin{proposition}\label{maxt2}
In any spanning tree $T$, $2+\sum_{i \ge 3}(i-2)\cdot t_i^T=t_1^T$.
\end{proposition}
Due to Proposition~\ref{maxt2}, \mist on subcubic graphs boils down to finding a spanning tree $T$ such that $t_2^T$ is maximum. Every internal vertex
of higher degree would also introduce additonal leaves.
\begin{lemma}{\cite{PriSlo2003}}\label{ind} An optimal solution
 $T_o$ to {\sc Max Internal Spanning Tree} is a Hamiltonian path or the leaves of $T_o$ are independent.
\end{lemma}
The proof of Lemma~\ref{ind} shows that if $T_o$ is not a Hamiltonian path and there are two adjacent leaves, then the number of internal vertices can be increased.
In the rest of the paper we assume that $T_o$ is not a Hamiltonian path due to the next lemma. 

\shortversion{Using the $\Oh^*(1.251^n)$ algorithm for \textsc{Hamiltonian Cycle} \cite{IwaNak07} we can easily prove the following.}

\begin{lemma}
\textsc{Hamiltonian Path} can be solved in time $\Oh^*(1.251^n)$ on subcubic graphs.
\end{lemma}
\longversion{
\begin{proof}
Let $G=(V,E)$ be a subcubic graph. Run the algorithm of \cite{IwaNak07} to find a Hamiltonian cycle. If it succeeds $G$ clearly also has a Hamiltonian
path. If it does 
not succeed we have to investigate if 
 $G$ has a Hamiltonian path whose end points are not adjacent. Let $u,v \in V(G)$ be two non-adjacent vertices. To check whether $G$ has a Hamiltonian
path $uPv$, we check
 whether $G'=(V,E'):=E\cup\{\{u,v\}\})$ has a Hamiltonian cycle. If $G'$ has maximum degree at most $3$, then run the algorithm of \cite{IwaNak07}.
Otherwise, choose a vertex of degree $4$, say $u$, and two neighbors $x,z$ of $u$ distinct from~$v$. As $\{u,v\}$ belongs to every Hamiltonian cycle
of $G'$ (otherwise $G$ has a Hamiltonian cycle too), every Hamiltonian cycle of $G'$ avoids $\{u,x\}$ or $\{u,z\}$. Recursively check if
$(V,E'\setminus\{\{u,x\}\})$ or $(V,E'\setminus\{\{u,z\}\})$ has a Hamiltonian cycle. This recursion has depth at most $2$ since $G'$ has at most $2$
vertices of degree $4$.
 The \longversion{\textsc{Hamiltonian Cycle }}algorithm of \cite{IwaNak07} is executed at most $4(n(n-1)/2-m)$ times.
 This algorithms runs in  $\longversion{\Oh^*(2^{(31/96)n})\subseteq}\Oh^*(1.2509^n)$ steps.
\qed
%
%
\end{proof}
}



\begin{lemma}\label{nodeg3}
Let $T$ be a spanning tree and $v \in V(T)$ with $d_T(v)=3$. Suppose there is a $u \in N(v)$ such that $d_T(u)=3$ and $\{u,v\}$ is not a bridge. Then there is a spanning tree $T'\supset (T\setminus\{\{u,v\}\})$ with $|I(T')|\ge |I(T)|$ and $d_{T'}(u)=d_{T'}(v)=2$.
\end{lemma}
\begin{proof}
By removing $\{u,v\}$, $T$ is separated into two parts $T_1$ and $T_2$. The vertices $u$ and $v$ become 2-vertices.  As $\{u,v\}$ is not a bridge, there is another edge $e \in E \setminus E(T)$ connecting $T_1$ and $T_2$. By adding $e$ we lose at most two 2-vertices. Then let $T':=(T \setminus \{\{u,v\}\}) \cup \{e\}$ and it follows that $|I(T')|\ge |I(T)|$.\qed
\end{proof}

\subsection{Reduction Rules}
Let $E' \subseteq E$. Then, $\partial E':=\{\{u,v\} \in E \setminus E' \mid u \in V(E')\}$ are the edges outside $E'$ that have a common end point with an edge in $E'$ and $\partial_V E':=V(\partial E') \cap V(E')$ are the vertices that have at least one incident edge in $E'$ and another incident edge not in $E'$.
In the course of the algorithm we will maintain an acyclic subset of edges $F$ which will be part of the final solution. The following invariant will always be true: $G[F]$ consists of a tree $T$ and a set $P$ of \emph{pending tree edges (pt-edges)}. Here a pt-edge $\{u,v\} \in F$ is an edge with one end point $u$ of degree $1$ and the other end point $v \not \in V(T)$. $G[T \cup P]$ will always consist of $1+|P|$ components. \\[1ex]
Next we present \longversion{a sequence of}\shortversion{several} reduction rules. \longversion{Note that t}\shortversion{T}he order in which they are applied is crucial\shortversion{: B}\longversion{. We assume that b}efore a rule is applied the preceding ones were carried out exhaustively.

\begin{figure}[bt]
\psfrag{a}{$a$}
\psfrag{b}{$b$}
\psfrag{u}{$u$}
\psfrag{v}{$v$}
\psfrag{w}{$w$}
\psfrag{x}{$x$}
\psfrag{z}{$z$}
\centering
\subfigure[ ]{\label{1b}
\includegraphics[scale=\scl]{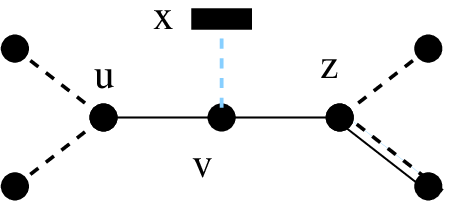}}\hspace*{3ex}
\subfigure[ ]{\label{1c}
\includegraphics[scale=\scl]{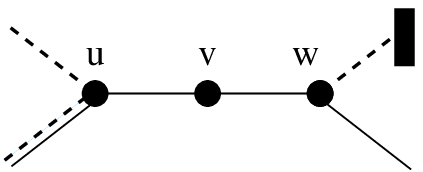}}
\caption{Light edges may be not present. Double edges (dotted or solid, resp.) refer to edges which are either $T$-edges or not, resp. Edges attached to oblongs are pt-edges.}
\label{attach}
\end{figure}

\begin{enumerate}\itemsep.5pt
\item \redrule{Cycle}{Delete any edge $e \in E$ such that $E(T)\cup\{e\}$ has a cycle.}
\item \redrule{Bridge}{If there is a bridge $e \in \partial E(T)$, then add $e$ to $F$. }
\item \redrule{Deg1}{If there is a degree-$1$ vertex $u \in V\setminus V(F)$, then add its incident edge to $F$. }
\item \redrule{Pending}{If there is a vertex $v$ that is incident to $d_G(v)-1$ pt-edges, then remove its incident pt-edges. }
\item \redrule{ConsDeg2}{If there are edges $\{v,w\},\{w,z\} \in E \setminus E(T)$ such that $d_G(w)=d_G(z)=2$, then delete $\{v,w\},\{w,z\}$ from $G$ and add the edge $\{v,z\}$ to $G$.}
\item \redrule{Deg2}{If there is an edge $\{u,v\} \in \partial E(T)$ such that $u \in V(T)$ and $d_G(u)=2$, then add $\{u,v\}$ to $F$.}
\item \redrule{Attach}{If there are edges $\{u,v\},\{v,z\} \in \partial E(T)$ such that $u,z \in V(T)$, $d_T(u)=2$, $1 \le d_T(z) \le 2$, then delete
$\{u,v\}$. .}
\item \redrule{Attach2}{If there is a vertex $u \in \partial_V E(T)$ with $d_T(u)=2$ and $\{u,v\} \in E \setminus  E(T)$ such that $v$ is incident to
a pt-edge, then delete $\{u,v\}$. See Fig.~\ref{1b}}
\item \redrule{Special}{If there are two edges $\{u,v\},\{v,w\} \in E \setminus F$ with  $d_T(u) \ge 1$, $d_G(v)=2$, and $w$ is incident to a pt-edge, then add $\{u,v\}$ to $F$ (Fig.~\ref{1c}). }
\end{enumerate}

\begin{lemma}
The reduction rules stated above are sound.
\end{lemma}
\begin{proof}
Let $T_o \supset F$ be an optimal spanning tree of $G$. 
The first three rules are correct for the purpose of connectedness and acyclicity of the evolving spanning tree. 
\begin{description}\itemsep.5pt
\item[\bf Pending] is correct as the other edge incident to $v$ (which will be added to $P$ by a subsequent {\bf Deg1} rule) is a bridge and needs to be in any spanning tree. 
\item [\bf ConsDeg2] We implicitly assume that we can add $\{w,z\}$ to 
 $T_o \supset F$. If $\{w,z\} \notin E(T_o)$ then $\{v,w\} \in E(T_o)$. Then we can simply exchange the two edges giving a solution $T'_o$ with $\{w,z\} \in E(T'_o)$ and $t_2^{T_o}\le t_2^{T_o'}$.  
\item[\bf Deg2] Since the preceding reduction rules do not apply, we have $d_G(v)=3$ and there is one edge, say
$\{v,z\}$, $z\neq u$, that is not pending. 
Assume $T_o$ 
has $u$ as a leaf. Define 
another spanning tree $T_o'\supset F$ by setting $T_o'=(T_o\cup\{\{u,v\}\})\setminus\{v,z\}$.
Since 
$|I(T_o)|\leq |I(T_o')|$, $T_o'$ is also optimal.   
\item[\bf Attach] If $\{u,v\} \in E(T_o)$ then $\{v,z\} \not \in E(T_o)$ due to the acyclicity of $T_o$ and as $T$ is connected. Then by exchanging $\{u,v\}$ and $\{v,z\}$ we obtain a solution $T_o'$ with at least as many 2-vertices.
\item[\bf Attach2] Suppose $\{u,v\} \in E(T_o)$. Let $\{v,p\}$ be the pt-edge and $\{v,z\}$ the third edge incident to $v$ (that must exist and is not pending, since {\bf Pending} did not apply). Since {\bf Bridge} did not apply, $\{u,v\}$ is not a bridge. 
Firstly, suppose $\{v,z\} \in E(T_o)$. Due to the proof of Lemma~\ref{nodeg3}, there is also an optimal solution $T'_o\supset F$ with $\{u,v\} \notin E(T'_o)$.
Secondly, assume $\{v,z\} \notin E(T_o)$. Then $T'=(T_o \setminus \{\{u,v\}\}) \cup \{\{v,z\}\}$ is also optimal as $u$ has become a 2-vertex.
\item[Special] Suppose $\{u,v\} \not \in E(T_o)$. Then $\{v,w\},\{w,z\} \in E(T_o)$ where $\{w,z\}$ is the third edge incident to $w$. Let $T_o':=(T_o \setminus \{\{v,w\}\}) \cup \{\{u,v\}\}$. In $T_o'$, $w$ is a 2-vertex and hence $T'$ is also optimal.\qed
\end{description}
\end{proof}

\subsection{The Algorithm}

The algorithm we describe here is recursive. It constructs a set $F$ of edges which are selected to be in every spanning tree considered in the current recursive step. The algorithm chooses edges and considers all relevant choices for adding them to $F$ or removing them from $G$. It selects these edges based on priorities chosen to optimize the running time analysis. Moreover, the set $F$ of edges will always be the union of a tree $T$ and a set of edges $P$ that are not incident to the tree and have one end point of degree $1$ in $G$ (pt-edges). We do not explicitly write in the algorithm that edges move from $P$ to $T$ whenever an edge is added to $F$ that is incident to both an edge of $T$ and an edge of $P$. To maintain the connectivity of $T$, the algorithm explores edges in the set $\partial E(T)$ to grow $T$.

If $|V|>2$ every spanning tree $T$ must have a vertex $v$ with $d_T(v) \ge 2$. Thus initially the algorithm creates an instance for every vertex $v$ and every possibility that $d_T(v) \ge 2$. Due to the degree constraint there are no more than $4n$ instances. After this initial phase, the algorithm proceeds as follows.

\begin{enumerate}
 \item[1.] Carry out each reduction rule exhaustively in the given order (until no rule applies).
 \item[2.] If $\partial E(T)=\emptyset$ and $V \neq V(T)$, then $G$ is not connected and does not admit a spanning tree. Ignore this branch.
 \item[3.] If $\partial E(T)=\emptyset$ and $V = V(T)$, then return $T$.
 \item[4.] Select $\{a,b\} \in \partial E(T)$ with $a \in V(T)$ according to the following priorities (if such an edge exists):
  \begin{itemize}
   \item[a)] there is an edge $\{b,c\} \in \partial E(T)$,
   \item[b)] $d_G(b)=2$,
   \item[c)] $b$ is incident to a pt-edge, or
   \item[d)] $d_T(a)=1$.
  \end{itemize}
  Recursively solve two instances where $\{a,b\}$ is added to $F$ or removed from $G$ respectively, and return a spanning tree with most internal vertices.
 \item[5.] Otherwise, select $\{a,b\} \in \partial E(T)$ with $a \in V(T)$. Let $c,x$ be the other two neighbors of $b$. Recursively solve three instances where 
  \begin{itemize}
   \item[(i)] $\{a,b\}$ is removed from $G$,
   \item[(ii)] $\{a,b\}$ and $\{b,c\}$ are added to $F$ and $\{b,x\}$ is removed from $G$, and
   \item[(iii)] $\{a,b\}$ and $\{b,x\}$ are added to $F$ and $\{b,c\}$ is removed from $G$.
  \end{itemize}
  Return a spanning tree with most internal vertices.
\end{enumerate}

\longversion{\subsection{An Exact Analysis of the Algorithm}}

By a Measure \& Conquer analysis taking into account the degrees of the vertices, their number of incident edges that are in $F$, and to some extent
the degrees of their neighbors, 
we obtain the following result.

\begin{theorem}\label{thm:exalg}
{\sc Max Internal Spanning Tree} can be solved in time $\Oh^*(\rt^n)$ on subcubic graphs.
\end{theorem}


\shortversion{The proof is omitted for reasons of space, but we will}\longversion{Let us} provide measure we used in the following:
Let $D_2:=\{v \in V \mid d_G(v)=2, d_T(v)=0\}$, $D_3^\ell:=\{v \in V \mid d_G(v)=3, d_T(v)=\ell\}$ and $D_3^{2\ast}:=\{v \in D_3^2 \mid N_G(v)
\setminus N_T(v)=\{u\} \text{ and } d_G(u)=d_T(u)=2\}$. 
Then the measure we use for our running time bound is:
$$
\mu (G)=\omega_2 \cdot |D_2|+ \omega_3^1 \cdot |D_3^1| + \omega_3^2 \cdot |D_3^2 \setminus D_3^{2\ast}| + |D_3^0| + \omega_3^{2\ast}\cdot|D_3^{2\ast}|
$$
with $\omega_2=0.3193$    $\omega_3^1=0.6234$,    $\omega_3^2=0.3094$ and     $\omega_3^{2\ast}=0.4144$.
\longversion{

Let $\Delta_3^0:=\Delta_3^{0\ast}:=1-\omega_3^{1}$, $\Delta_3^1:=\omega_3^1-\omega_3^{2}$, $\Delta_3^{1\ast}:=\omega_3^1-\omega_3^{2\ast}$, $\Delta_3^2:=\omega_3^2$, $\Delta_3^{2\ast}:=\omega_3^{2 \ast}$ and $\Delta_2=1-\omega_2$. We define $\tilde{\Delta}_3^{i}:= \min\{\Delta_3^i,\Delta_3^{i\ast}\}$ for $1 \le i \le 2$, 
$\Delta_m^{\ell}=\min_{0 \le j \le \ell} \{\Delta_3^j\}$, $\tilde \Delta_m^{\ell}=\min_{0 \le j \le \ell} \{\tilde \Delta_3^j\}$.}
\longversion{
 }
The proof of the theorem is using the following result:

\begin{lemma}\label{noinc}
None of the reduction rules increase $\mu$ for the given weights.
\end{lemma}

\longversion{
\begin{proof}
{\bf Bridge}, {\bf Deg1}, {\bf Deg2} and {\bf Special} add edges to $T$. Due to the definitions of $D_3^\ell$ and $D_3^{2\ast}$ and the
 choice of the weights it can be seen that $\mu$ only decreases. It is also easy to see that the deletion of edges $\{u,v\}$ with $d_T(u)\ge 1$ is
safe with respect to $u$. The weight of $u$ can only decrease due to this. Nevertheless, the rules which delete edges might cause that a $v \in
D_3^{2}\setminus D_3^{2\ast}$ will be in $D_3^{2\ast}$ afterwards. Thus, we have to prove that in this case the overall reduction is enough. A
sufficient criterion that the described scenario takes place is if degree 2 vertices are created. {\bf Cycle} may create vertices of degree 2, but
none which are adjacent to a vertex in $D_3^{2}\setminus D_3^{2\ast}$ and are not subject to another application of {\bf Cycle}. The next reduction
rule which may create vertices of degree 2 is {\bf Attach} when $d(v)=2$. The minimum reduction is $\omega_3^2+ \Delta_2
-(\omega_3^{2\ast}-\omega_3^2)>0$. No other reduction rule creates degree 2 vertices.\qed
\end{proof}

\begin{proof} (Theorem~\ref{thm:exalg}) 
As the algorithm deletes edges or moves edges from $E\setminus F$ to $F$, cases 1--3 do not contribute to the exponential function in the running time of the algorithm. It remains to analyze cases 4 and 5, which we do now. Note that after applying the reduction rules exhaustively, we have that for all $v \in \partial_V E(T)$, $d_G(v)=3$ ({\bf Deg2}) and for all $u \in V$, $d_P(u) \le 1$ ({\bf Pending}).

\begin{enumerate}\itemsep.5pt
 \item[\bf 4.(a)] Obviously, $\{a,b\},\{b,c\} \in E\setminus E(T)$, and there is a vertex $d$ such that $\{c,d\} \in E(T)$; see Figure~\ref{branch1}. We must have $d_T(a)=d_T(c)=1$ (due to the reduction rule {\bf Attach}).
We consider three cases.
\begin{itemize}
\item $d_G(b)=2$. When $\{a,b\}$ is added to $F$, {\bf Cycle} deletes $\{b,c\}$. We get an amount of $\omega_2$ and $\omega_3^1$ as $b$ drops out of $D_2$ and $c$ out of $D_3^1$ ({\bf Deg2}). Also $a$ will be removed from $D_3^1$ and added to $D_3^2$ which amounts to a reduction of at least $\tilde{\Delta}_3^1$. When $\{a,b\}$ is deleted, $\{b,c\}$ is added to $E(T)$ ({\bf Bridge}). By a symmetric argument we get a reduction of $\omega_2+\omega_3^1+\tilde \Delta_3^1$ as well. In total this yields a $(\omega_2+\omega_3^1+\tilde \Delta_3^1,\omega_2+\omega_3^1+\tilde \Delta_3^1)$-branch.
\item $d_G(b)=3$ and there is one  pt-edge attached to $b$. Adding $\{a,b\}$ to $F$ decreases the measure by $\tilde \Delta_3^1$ (from $a$) and $2\omega_3^1$ (deleting $\{b,c\}$, then {\bf Deg2} on $c$). By Deleting $\{a,b\}$  we decrease $\mu$ by $2\omega_3^1$ and by $\tilde \Delta_3^1$ (from $c$).   This amounts to a $(2\omega_3^1+\tilde \Delta_3^1,2\omega_3^1+\tilde \Delta_3^1)$-branch.
\item $d_G(b)=3$ and no pt-edge is attached to $b$. Let $\{b,z\}$ be the third edge incident to $b$. In the first branch the measure drops by at least $\omega_3^1+\tilde \Delta_3^1$ from $c$ and $a$ ({\bf Deg2}), $1$ from $b$ ({\bf Deg2}). 
In the second branch we get $\omega_3^1+\Delta_2$. Observe that we also get an amount of at least $\tilde \Delta_m^1$ from $q \in N_T(a) \setminus \{b\}$ if $d_G(q)=3$. If $d_G(q)=2$ we get $\omega_2$.  It results 
a $(\omega_3^1+\tilde \Delta_3^1+1,\omega_3^1+\Delta_2+min\{\omega_2,\tilde \Delta_m^1\})$-branch.
\end{itemize}
Note that from this point on, for all $u,v \in V(T)$ there is no $z \in V\setminus V(T)$ with $\{u,z\},\{z,v\}\in E$.\vspace*{2ex}

\item[\bf 4.(b)] As the previous case does not apply, the other neighbor $c$ of $b$ has $d_T(c)=0$, and $d_G(c)\ge 2$ ({\bf Pending}), see Figure~\ref{branch1.5}. Additionally, observe that we must have $d_G(c)=3$ ({\bf ConsDeg2}) and that $d_P(c)=0$ due to {\bf Special}. We consider two subcases.
\begin{itemize}
\item[$I)$] $d_T(a)=1$. When we add $\{a,b\}$ to $F$, then $\{b,c\}$ is also added due to {\bf Deg2}. The reduction is at least $\tilde \Delta_3^1$
from $a$, $\omega_2$
 from $b$ and $\Delta_3^0$ from $c$. When $\{a,b\}$ is deleted, $\{b,c\}$ becomes a pt-edge. There is $\{a,z\} \in E \setminus E(T)$ with $z \neq b$,
which is subject to a {\bf Deg2} reduction rule. We get at least $\omega_3^1$ from $a$, $\omega_2$ from $b$, $\Delta_3^0$ from $c$ and
$\min\{\omega_2,\tilde \Delta_m^1\}$ from $z$. This is a $(\tilde \Delta_3^1+\Delta_3^0+\omega_2,\omega_3^1+\Delta_3^0+\omega_2+\min\{\omega_2,\tilde
\Delta_m^1\})$-branch. 
\item[$II)$] $d_T(a)=2$. Similarly, we obtain a $(\Delta_3^{2\ast}+\omega_2+\Delta_3^0,\Delta_3^{2\ast}+\omega_2+\Delta_3^0)$-branch.
\end{itemize}
\item[\bf 4.(c)] In this case, $d_G(b)=3$ and there is one  pt-edge attached to $b$, see Figure~\ref{branch2}. 
Note that $d_T(a)=2$ can be ruled out due to {\bf Attach2}. Thus, $d_T(a)=1$. Let $z \neq b$ be such that $\{a,z\} \in E \setminus E(T)$. Due to the priorities, $d_G(z)=3$.  We distinguish between the cases where $c$ is incident to a pt-edge or not.
\begin{enumerate}
 \item $d_P(c)=0$. First suppose $d_G(c)=3$. Adding $\{a,b\}$ to $F$ allows a reduction of $2\Delta_3^1$ (due to case 4.(b) we can exclude
$\Delta_3^{1\ast}$). Deleting $\{a,b\}$ 
implies that we get a reduction from $a$ and $b$ of $2\omega_3^1$ ({\bf Deg2} and {\bf Pending}). As $\{a,z\}$ is added to $F$ we reduce $\mu (G)$ by
at least $\tilde \Delta_3^1$ as the state of $z$ changes. 
 Now due to {\bf Pending} and {\bf Deg1} we include $\{b,c\}$ and get $\Delta_3^0$ from $c$. We have at least a $(2\Delta_3^1,2\omega_3^1+\tilde\Delta_3^1+\Delta_3^0)$-branch.\\
If $d_G(c)=2$ we consider the two cases for $z$ also. These are $d_P(z)=1$ and $d_P(z)=0$. The first  entails $(\omega_3^1+
\Delta_3^{1\ast},2\omega_3^1+ \tilde \Delta_3^1+\omega_2+\tilde \Delta_m^2)$.
 Note that when we add $\{a,b\}$ we trigger {\bf Attach2}. The second is a $(\Delta_3^1+ \Delta_3^{1\ast},2\omega_3^1+\Delta_3^0+\omega_2+\tilde
\Delta_m^2)$-branch.
\item $d_P(c)=1$. Let $d \neq b$ be the other neighbor of $c$ that does not have degree $1$. When $\{a,b\}$ is added to $F$, $\{b,c\}$ is deleted by
{\bf Attach2} and $\{c,d\}$ becomes a pt-edge 
({\bf Pending} and {\bf Deg1}). The changes on $a$ incur a measure decrease of $\Delta_3^{1\ast}$ and those on $b,c$ a measure decrease of $2
\omega_3^1$.
When $\{a,b\}$ is deleted, $\{a,z\}$ is added to $F$ ({\bf Deg2}) and $\{c,d\}$ becomes a pt-edge by two applications of the {\bf Pending} and {\bf
Deg1} rules. Thus,
 the decrease of the measure is at least $3 \omega_3^1$ in this branch. In total, we have a
$(\Delta_3^{1\ast}+2\omega_3^1,3\omega_3^1)$-branch here.
\end{enumerate}
\item[\bf 4.(d)] Now, $d_G(b)=3$, $b$ is not incident to a pt-edge, and $d_T(a)=1$. See Figure~\ref{branch2}. There is also some $\{a,z\} \in E\setminus E(T)$ such that $z \ne b$.
Note that $d_T(z)=0$, $d_G(z)=3$ and $d_P(z)=0$. Otherwise either {\bf Cycle} or  cases 4.(b) or 4.(c) would have been triggered. 
From the addition of $\{a,b\}$ to $F$ we get $\Delta_3^1+\Delta_3^0$ and from its deletion $\omega_3^1$ (from $a$ via {\bf Deg2}), $\Delta_2$ (from
$b$) and at least $\Delta_3^0$ from $z$ and thus, a $(\Delta_3^1+\Delta_3^0,\omega_3^1+\Delta_2+\Delta_3^0)$-branch.

\item[\bf 5.] See Figure~\ref{branch3}. The algorithm branches in the following way: $1)$ Delete $\{a,b\}$, $2)$ add $\{a,b\},\{b,c\}$, and delete $\{b,x\}$, $3)$ add $\{a,b\},\{b,x\}$ and delete $\{b,c\}$. Due to 
{\bf Deg2}, 
we can disregard the case when $b$ is a leaf. Due to Lemma~\ref{nodeg3} we also disregard the case when $b$ is a 3-vertex. Thus by branching in this manner we find at least one optimal solution.\\[0.5ex]
The reduction in the first branch is at least $\omega_3^2+\Delta_2$. We get an additional amount of $\omega_2$ if $d(x)=2$ or $d(c)=2$ from {\bf ConsDeg2}. In the second we have to consider also the vertices $c$ and $x$. There are exactly three situations for $h \in \{c,x\}$ $\alpha)$ $d_G(h)=2$, $\beta)$ $d_G(h)=3$, $d_P(h)=0$ and $\gamma)$ $d_G(h)=3$, $d_P(h)=1$. We will only analyze branch $2)$ as $3)$ is symmetric.  We first get a reduction of $\omega_3^2+1$ from $a$ and $b$.
We reduce $\mu$ due to deleting $\{b,x\}$ by: $\alpha)$ $\omega_2+\tilde \Delta_m^2$, $\beta)$ $\Delta_2$, $\gamma)$ $\omega_3^1+\tilde \Delta_m^2$. 
Next we examine the amount by which $\mu$ will be decreased by adding $\{b,c\}$ to $F$. We distinguish between the cases $\alpha,\beta$ and $\gamma$: $\alpha)$ $\omega_2+\tilde \Delta_m^2$, $\beta)$ $\Delta_3^0$, $\gamma)$ $\tilde \Delta_3^1$.\\
For $h \in\{c,x\}$ and $W \in \{\alpha,\beta,\gamma\}$ let $1_W^h$ be the indicator function which is set to one if we have situation $W$ at vertex $h$. Otherwise it is zero. Now the branching tuple can be stated the following way :\\
$(\omega_3^2+\Delta_2+(1_{\alpha}^x+1_\alpha^c)\cdot \omega_2,$
$\omega_3^2+1+1^x_\alpha \cdot (\omega_2+\tilde \Delta_m^2)+1^x_\beta \cdot \Delta_2 +1^x_\gamma \cdot(\omega_3^1+\tilde \Delta_m^2)+1^c_\alpha\cdot (\omega_2+ \tilde \Delta_m^2)+ 1^c_\beta \cdot \Delta_3^0+1^c_\gamma \cdot \tilde \Delta_3^1),$\\
$ \omega_3^2+1+1^c_\alpha \cdot (\omega_2+\tilde \Delta_m^2)+1^c_\beta \cdot \Delta_2 +1^c_\gamma \cdot(\omega_3^1+\tilde \Delta_m^2)+1^x_\alpha\cdot (\omega_2+\tilde \Delta_m^2) + 1^x_\beta \cdot \Delta_3^0+1^x_\gamma \cdot \tilde\Delta_3^1)$\\
The amount of $(1_{\alpha}^x+1_\alpha^c)\cdot \omega_2$ comes from possible applications of {\bf ConsDeg2}.
\end{enumerate}


\begin{figure}\centering
\psfrag{a}{$a$}
\psfrag{b}{$b$}
\psfrag{b}{$b$}
\psfrag{z}{$z$}
\psfrag{x}{$x$}
\psfrag{y}{$y$}
\subfigure[\label{branch1}]{\includegraphics[scale=\scl]{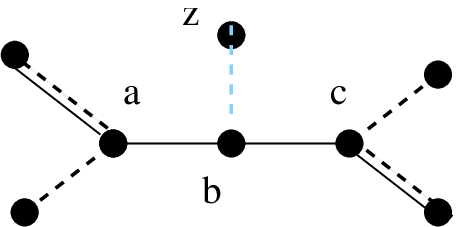}}\hspace*{0.75ex}
\subfigure[\label{branch1.5}]{\includegraphics[scale=\scl]{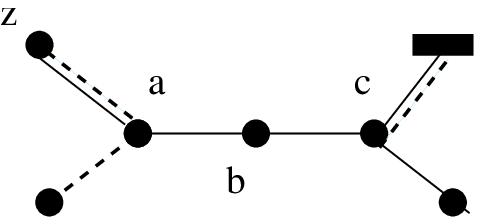}}\hspace*{0.75ex}
\subfigure[\label{branch2}]{\includegraphics[scale=\scl]{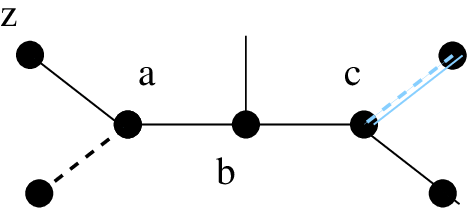}}\hspace*{0.75ex}
\subfigure[\label{branch3}]{\includegraphics[scale=\scl]{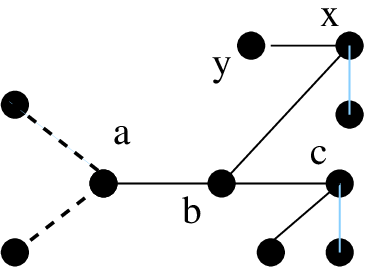}}
\caption{}
\label{branch}
\end{figure}

Observe that every instance created by branching is smaller than the original instance in terms of $\mu$. Together with Lemma~\ref{noinc} we see that every step of the algorithm only decreases $\mu$.
Now if we evaluate the upper bound for every given branching tuple for the given weights we can conclude that {\sc Max Internal Spanning Tree} can be
solved in time $\Oh^*(\rt^n)$ on subcubic graphs.
\qed
\end{proof}
}

\longversion{\subsection{A Parameterized Analysis of the Algorithm}}

For general graphs, the smallest known kernel has size $3k$. This can be easily improved to $2k$ for subcubic graphs.

\begin{lemma}\label{lem:kernel}
MIST on subcubic graphs has a $2k$-kernel.
\end{lemma}
\longversion{
 \begin{proof}
 Compute an arbitrary spanning tree $T$. If it has at least $k$ inner vertices, answer Yes. Otherwise, $t^T_3 +t^T_2 < k$. Then, by Proposition~\ref{maxt2}, $ t^T_1 < k+2$. Thus, $|V| \le 2k$. \qed
 \end{proof}
}

Applying the algorithm of Theorem~\ref{thm:exalg} on this kernel for subcubic graphs shows the following result.

\begin{corollary}
Deciding whether a subcubic graph has a spanning tree with at least $k$ internal vertices can be done in time $\rtk^k n^{\Oh(1)}$.
\end{corollary}

However, we can achieve a faster parameterized running time by applying a Measure \& Conquer analysis which is customized to the parameter $k$. We would like to put forward that our use of the technique of Measure \& Conquer for a parameterized algorithm analysis goes beyond previous work as our measure is not restricted to differ from the parameter $k$ by just a constant. We first demonstrate our idea with a simple analysis.

\begin{theorem}
Deciding whether a subcubic graph has a spanning tree with at least $k$ internal vertices can be done in time $\rtki^k n^{\Oh(1)}$.
\end{theorem}

\begin{proof}
 Consider the algorithm described earlier, with the only modification that the parameter $k$ is adjusted whenever necessary (for example, when two pt-edges incident to the same vertex are removed), and that the algorithm stops and answers Yes whenever $T$ has at least $k$ internal vertices. Note that the assumption that $G$ has no Hamiltonian path can still be made due to the $2k$-kernel of Lemma~\ref{lem:kernel}: the running time of the Hamiltonian path algorithm is $1.251^{2k}n^{\Oh(1)}=1.5651^k n^{\Oh(1)}$.
 The running time analysis of our algorithm relies on the following measure:
 \[
  \kappa := \kappa(G,F,k) := k - \omega \cdot |X| - |Y|,
 \]
 where $X:=\{v\in V \mid d_G(v)=3, d_T(v)=2\}$, $Y:=\{v\in V \mid d_G(v)=d_T(v)\ge 2\}$ and $0 \le \omega \le 1$. Let $U:=V\setminus (X \cup Y)$.
 Note that a vertex which has already been decided to be internal, but that still has an incident edge in $E \setminus T$, contributes a weight of
$1-\omega$ to the measure.
Or equivalently, such a vertex has  been only counted by a fraction of $\omega$.
 
 None of the reduction and branching rules increases $\kappa$ and we have that $0 \le \kappa \le k$ at any time of the execution of the algorithm.
 
\longversion{
 In step 4, whenever the algorithm branches on an edge $\{a,b\}$ such that $d_T(a)=1$ (w.l.o.g., we assume  \longversion{that }$a\in V(T)$), the measure decreases by at least $\omega$ in one branch, and by at least $1$ in the other branch. We speak of a $(\omega,1)$-branch. To see this, it suffices to look at vertex $a$. Due to \textbf{Deg2}, $d_G(a)=3$. When $\{a,b\}$ is added to $F$, vertex $a$ moves from the set $U$ to the set $X$. When $\{a,b\}$ is removed from $G$, a subsequent application of the {\bf Deg2} rule adds the other edge incident to $a$ to $F$, and thus, $a$ moves from $U$ to~$Y$.
 
 Still in step 4, let us consider the case where $d_T(a)=2$. Then condition (b) ($d_G(b)=2$) of step 4 must hold, due to the preference of the reduction and branching rules: condition (a) is excluded due to reduction rule \textbf{Attach}, (c) is excluded due to \textbf{Attach2} and (d) is excluded due to its condition that $d_T(a)=1$. When $\{a,b\}$ is added to $F$, the other edge incident to $b$ is also added to $F$ by a subsequent {\bf Deg2} rule. Thus, $a$ moves from $X$ to $Y$ and $b$ from $U$ to $Y$ for a measure decrease of $(1-\omega)+1=2-\omega$. When $\{a,b\}$ is removed from $G$, $a$ moves from $X$ to $Y$ for a measure decrease of $1-\omega$. Thus, we have a $(2-\omega,1-\omega)$-branch.
 
 In step 5, $d_T(a)=2$, $d_G(b)=3$, and $d_F(b)=0$. Vertex $a$ moves from $X$ to $Y$ in each branch and $b$ moves from $U$ to $Y$ in the two latter branches. In total we have a $(1-\omega,2-\omega,2-\omega)$-branch.
}

 By \shortversion{a simple case analysis, }setting $\omega = 0.45346$ and evaluating the branching factors, the proof follows.
\qed
\end{proof}

This analysis can be improved by also measuring the vertices of degree $2$ that are not adjacent to vertices of $X \cup Y$ and the vertices incident
to pt-edges differently.
 \shortversion{For reasons of space, we only exhibit the measure we used.}

\begin{theorem}\label{thm:rtkii}
Deciding whether a subcubic graph has a spanning tree with at least $k$ internal vertices can be done in time $\rtkii^k n^{\Oh(1)}$.
\end{theorem}

 The proof of this theorem follows the same lines as the previous one, except that we consider a more detailed measure:
 \[
  \kappa := \kappa(G,F,k) := k - \omega_1 \cdot |X| - |Y| - \omega_2 |Z| - \omega_3 |W|,
 \]
 where
 \begin{itemize}
  \item $X:=\{v\in V \mid d_G(v)=3, d_T(v)=2\}$ is the set of vertices of degree $3$ that are incident to exactly 2 edges of $T$,
  \item $Y:=\{v\in V \mid d_G(v)=d_T(v)\ge 2\}$ is the set of vertices of degree at least 2 that are incident to only edges of $T$,
  \item $W:=\{v\in V\setminus (X \cup Y) \mid d_G(v) \ge 2, \exists u\in N(v) \text{ st. } d_G(u)=d_F(u)=1\}$
is the set of vertices of degree at least $2$ that have an incident pt-edge, and
  \item $Z:=\{v\in V\setminus W \mid d_G(v)=2, N[v] \cap (X \cup Y) = \emptyset\}$ is the set of degree $2$ vertices that do not have a vertex of $X\cup Y$ in their closed neighborhood, and are not incident to a pt-edge.  
\end{itemize}
 We immediately set $\omega_1 :=  0.5485, \omega_2 :=  0.4189$ and $\omega_3 := 0.7712$. 
Let $U:=V\setminus (X \cup Y \cup Z \cup W)$.
We first have to show that the algorithm can be stopped whenever the measure drops to $0$ or less.
\begin{lemma}
 Let $G=(V,E)$ be a connected graph, $k$ be an integer and $F\subseteq E$ be a set of edges that can be partitioned into a tree $T$ and a set of pending edges $P$. If none of the reduction rules applies to this instance and $\kappa(G,F,k) \le 0$, then $G$ has a spanning tree $T^* \supseteq F$ with at least $k$ internal nodes.
\end{lemma}

\longversion{
\begin{proof}
 Since the vertices in $X\cup Y$ are internal in any spanning tree containing $F$, it is sufficient to show that there exists a spanning tree $T^* \supseteq F$ that has at least $\omega_2 |Z|+\omega_3 |W|$ more internal vertices than $T$.
 
 The spanning tree $T^*$ is constructed as follows. Greedily add a subset of edges $A \subseteq E\setminus F$ to $F$ to obtain a spanning tree $T'$ of
$G$. While there exists $v\in Z$ with neighbors $u_1$ and $u_2$ such that $d_{T'}(v)=d_{T'}(u_1)=1$ and $d_{T'}(u_2)=3$, then set $A:=(A \setminus
\{v,u_2\}) \cup \{u_1,v\}$. This procedure finishes in polynomial time as the number of 
internal vertices increases each time such a vertex is found. Call the resulting spanning tree $T^*$.

 By connectivity of a spanning tree, we have:
\begin{fact}\label{fact1}
 If $v \in W$, then $v$ is internal in $T^*$.
\end{fact}
 Note that $F \subseteq T^*$ as no vertex of $Z$ is incident to an edge of $F$. By the construction of $T^*$, we have the following.
\begin{fact}\label{fact3}
 If $u,v$ are two adjacent vertices in $G$ but not in $T^\ast$, such that $v\in Z$ and $u,v$ are leafs in $T^*$, then $v$'s other neighbor has
$T^*$-degree $2$.
\end{fact}

Let $Z_\ell \subseteq Z$ be the subset of vertices of $Z$ that are leafs in $T^*$ and let $Z_i := Z \setminus Z_\ell$. As $F \subseteq T^*$ and by
Fact~\ref{fact1}, 
all vertices of $X\cup Y\cup W\cup Z_i$ are internal in $T^*$.
Let $P$ denote the subset of vertices of $N(Z_\ell)$ that are internal in $T^*$. As $P$ might intersect with $W$ and for $u,v\in Z_\ell$, $N(u)$ 
and $N(v)$ might intersect (but $u \not \in N(v)$ because of \textbf{ConsDeg2}), we assign an initial potential of 1 to vertices of $P$. By
definition, $P \cap (X\cup Y) = \emptyset$. Thus the number of internal vertices in $T^*$ is at least $|X|+|Y|+|Z_i|+|P \cup W|$. To finish the proof
of the claim, we show that $|P \cup W|\ge \omega_2|Z_l|+\omega_3|W|$. 
 
 Decrease the potential of each vertex in $P \cap W$ by $\omega_3$. Then, for each vertex $v \in Z_\ell$, decrease the potential of each vertex in
$P_v=N(v)\cap P$ 
by $\omega_2/|P_v|$. We show that the potential of each vertex in $P$ remains positive.
 Let $u \in P$ and $v_1 \in Z_\ell$ be a neighbor of $u$. Note that $d_{T^*}(v_1)=1$. We disinguish two cases based on $u$'s tree-degree 
in $T^*$. If $d_{T^*}(u)=2$, then $u \not \in W$, as $u$ being incident to a pt-edge would contradict the connectivity of $T^*$. Moreover, $u$ is
incident to at most $2$ vertices of $Z_\ell$ (again by connectivity of $T^*$), its potential remains thus positive as $1-2 \omega_2 \ge 0$. If
$d_{T^*}(u)=3$ and $u \in W$ is incident to a pt-edge, then it has one neighbor in $Z_\ell$ (connectivity of $T^*$), which has only internal neighbors
(by Fact~\ref{fact3}). The potential of $u$ is thus $1-\omega_3-\omega_2/2 \ge 0$. If $d_{T^*}(u)=3$ and $u \not \in W$, then $u$ has at most two
neighbors in $Z_\ell$, and both of them have only inner neighbors due to Fact~\ref{fact3}. 
As $1-2 \omega_2 / 2 \ge 0$, $u$'s potential remains positive.
\qed
\end{proof}
}

We also show that reducing an instance does not increase its measure.
\begin{lemma}
 Let $(G',F',k')$ be an instance resulting from the application of a reduction rule to an instance $(G,F,k)$. Then, $\kappa(G',F',k') \le \kappa(G,F,k)$.
\end{lemma}
\longversion{
\begin{proof}
 If the reduction rule \textbf{Cycle} or \textbf{Attach2} is applied to $(G,F,k)$, then an edge in $\partial E(T)$ is removed from the graph. Then,

generalization

                                                                             analyze the parameter $k$ stays the same, and either each vertex remains
in the same set among $X,Y,Z,W,U$, or one or two vertices move from $X$ to $Y$, which we denote shortly by the status change of a vertex $u$: $\{X\}
\rightarrow \{Y\}$. The value of this status change is $(-1) - (-\omega_1) \le 0$. As the value of the status change is non-positive, it does not
increase the measure. From now on, we only write down the status changes, and implicitly check that their value is non-positive.
 
 If \textbf{Bridge} is applied, then let $e=\{u,v\}$ with $u\in \partial_V E(T)$. Vertex $u$ is either in $U$ or in $X$, and $v \in U \cup Z \cup W$. If $v \in U$, then $v \in U$ after the application of \textbf{Bridge}, as $v$ is not incident to an edge of $T$ (otherwise reduction rule \textbf{Cycle} would have applied). In this case, it is sufficient to check how the status of $u$ can change, which is $\{U\} \rightarrow \{Y\}$ if $u$ has degree $2$, $\{U\} \rightarrow \{X\}$ if $d_G(u)=3$ and $d_T(u)=1$, and $\{X\} \rightarrow \{Y\}$ if $d_G(u)=3$ and $d_T(u)=2$.
 If $v \in Z$, then $v$ moves to $U$ as $u$ necessarily ends up in $X \cup Y$. The possible status changes are $\{U,Z\} \rightarrow \{Y,U\}$ if $d_G(u)=2$, $\{U,Z\} \rightarrow \{X,U\}$, if $d_G(u)=3$ and $d_T(u)=1$, and $\{X,Z\} \rightarrow \{Y,U\}$ if $d_G(u)=3$ and $d_T(u)=2$.
 If $v \in W$, $v$ ends up in $X$ or $Y$, depending on whether it is incident to one or two pt-edges. The possible status changes are then $\{U,W\} \rightarrow \{Y,X\}$, $\{U,W\} \rightarrow \{Y,Y\}$, $\{U,W\} \rightarrow \{X,X\}$, $\{U,W\} \rightarrow \{X,Y\}$, $\{X,W\} \rightarrow \{Y,X\}$, and $\{X,W\} \rightarrow \{Y,Y\}$.
 
 If \textbf{Deg1} applies, the possible status changes are $\{X\} \rightarrow \{Y\}$, $\{U\} \rightarrow \{X\}$, $\{U\} \rightarrow \{W\}$, $\{U\} \rightarrow \{Y\}$, and $\{Z\} \rightarrow \{W\}$.
 
 In \textbf{Pending}, the status change $\{W\} \rightarrow \{U\}$ has negative value, but the measure still decreases as $k$ also decreases by $1$.
 
 Similarly, in \textbf{ConsDeg2}, a vertex in $Z \cup U$ disappears, but $k$ decreases by $1$.
 
 In \textbf{Deg2}, the possible status changes are $\{U\} \rightarrow \{Y\}$, $\{U,Z\} \rightarrow \{Y,U\}$, and $\{U,W\} \rightarrow \{Y,X\}$.
 
 In \textbf{Attach}, $u$ moves from $X$ to $Y$. Thus the status change $\{X\} \rightarrow \{Y\}$.
 
 Finally, in \textbf{Special}, the possible status changes are $\{U,Z\} \rightarrow \{X,U\}$ and $\{X\} \rightarrow \{Y\}$.
\qed
\end{proof}
}

\longversion{
\begin{proof} (of Theorem~\ref{thm:rtkii}) 
Table~\ref{tab:paramCA} outlines how vertices $a$, $b$, and their neighbors move between $U$, $X$, $Y$, $Z$, and $W$ in the branches where an edge is added to $F$ or deleted from $G$ in the different cases of the algorithm. For each case, the worst branching tuple is given.

\begin{table}[htbp]
\noindent
\begin{tabular*}{\textwidth}{l c c c c}
\hline
& add & delete & branching tuple\\
\hline
\multicolumn{4}{l}{Case 4.(a), $d_G(b)=2$}\\
\hline
\multirow{4}{*}{\begin{tikzpicture}[scale=0.66]
 \tikzstyle{vertex}=[minimum size=2mm,circle,fill=black,inner sep=0mm]
 \draw (0,0) node[vertex,label=below:$a$] (a) {};
 \draw (1,0) node[vertex,label=below:$b$] (b) {};
 \draw (2,0) node[vertex,label=below:$c$] (c) {};
 \draw (-1,1) node (a1) {};
 \draw (-1,-1) node (a2) {};
 \draw (3,1) node (c1) {};
 \draw (3,-1) node (c2) {};
 \draw (a2)--(a)--(b)--(c)--(c2);
 \draw[dashed] (a1)--(a) (c)--(c1);
\end{tikzpicture}} &
$a: U\rightarrow X$ & \multirow{3}{*}{symmetric} & \multirow{3}{*}{$(1+\omega_1-\omega_2,1+ \omega_1-\omega_2)$}\\
& $b: Z\rightarrow U$ & &\\
& $c: U\rightarrow Y$ & &\\
& & &\\
\hline
\multicolumn{4}{l}{Case 4.(a), $d_G(b)=3$, $b$ is incident to a pt-edge}\\
\hline
\multirow{4}{*}{\begin{tikzpicture}[scale=0.66]
 \tikzstyle{vertex}=[minimum size=2mm,circle,fill=black,inner sep=0mm]
 \tikzstyle{ptvertex}=[minimum size=2mm,rectangle,minimum width=4mm,fill=black,inner sep=0mm]
 \draw (0,0) node[vertex,label=below:$a$] (a) {};
 \draw (1,0) node[vertex,label=below:$b$] (b) {};
 \draw (2,0) node[vertex,label=below:$c$] (c) {};
 \draw (-1,1) node (a1) {};
 \draw (-1,-1) node (a2) {};
 \draw (3,1) node (c1) {};
 \draw (3,-1) node (c2) {};
 \draw (1,1) node[ptvertex] (b1) {};
 \draw (a2)--(a)--(b)--(c)--(c2);
 \draw[dashed] (a1)--(a) (c)--(c1) (b)--(b1);
\end{tikzpicture}} &
$a: U\rightarrow X$ & \multirow{3}{*}{symmetric} & \multirow{3}{*}{$(2+\omega_1-\omega_3, 2+\omega_1-\omega_3)$}\\
&$b: W\rightarrow Y$ & &\\
&$c: U\rightarrow Y$ & &\\
& & &\\
\hline
\multicolumn{4}{l}{Case 4.(a), $d_G(b)=3$, $b$ is not incident to a pt-edge}\\
\hline
\multirow{4}{*}{\begin{tikzpicture}[scale=0.66]
 \tikzstyle{vertex}=[minimum size=2mm,circle,fill=black,inner sep=0mm]
 \tikzstyle{ptvertex}=[minimum size=2mm,rectangle,minimum width=4mm,fill=black,inner sep=0mm]
 \draw (0,0) node[vertex,label=below:$a$] (a) {};
 \draw (1,0) node[vertex,label=below:$b$] (b) {};
 \draw (2,0) node[vertex,label=below:$c$] (c) {};
 \draw (-1,1) node (a1) {};
 \draw (-1,-1) node (a2) {};
 \draw (3,1) node (c1) {};
 \draw (3,-1) node (c2) {};
 \draw (1,1) node (b1) {};
 \draw (a2)--(a)--(b)--(c)--(c2) (b)--(b1);
 \draw[dashed] (a1)--(a) (c)--(c1);
\end{tikzpicture}} &
$a: U\rightarrow X$ & $a:U\rightarrow Y$ & \multirow{3}{*}{$(2+\omega_1, 1+\omega_2)$}\\
&$b: U\rightarrow Y$ & $b:U\rightarrow Z$ &\\
&$c: U\rightarrow Y$ & &\\
& & &\\
\hline
\multicolumn{4}{l}{Case 4.(b), $d_T(a)=1$}\\
\hline
\multirow{4}{*}{\begin{tikzpicture}[scale=0.66]
 \tikzstyle{vertex}=[minimum size=2mm,circle,fill=black,inner sep=0mm]
 \draw (0,0) node[vertex,label=below:$a$] (a) {};
 \draw (1,0) node[vertex,label=below:$b$] (b) {};
 \draw (2,0) node[vertex,label=below:$c$] (c) {};
 \draw (-1,1) node (a1) {};
 \draw (-1,-1) node (a2) {};
 \draw (3,1) node (c1) {};
 \draw (3,-1) node (c2) {};
 \draw (a2)--(a)--(b)--(c)--(c2) (c)--(c1);
 \draw[dashed] (a1)--(a);
\end{tikzpicture}} &
$a: U\rightarrow X$ & $a:U\rightarrow Y$ & \multirow{3}{*}{$(1+\omega_1-\omega_2, 1+\omega_3-\omega_2)$}\\
&$b: Z\rightarrow Y$ & $b:Z\rightarrow U$ &\\
& & $c:U\rightarrow W$ &\\
& & &\\
\hline
\multicolumn{4}{l}{Case 4.(b), $d_T(a)=2$}\\
\hline
\multirow{4}{*}{\begin{tikzpicture}[scale=0.66]
 \tikzstyle{vertex}=[minimum size=2mm,circle,fill=black,inner sep=0mm]
 \draw (0,0) node[vertex,label=below:$a$] (a) {};
 \draw (1,0) node[vertex,label=below:$b$] (b) {};
 \draw (2,0) node[vertex,label=below:$c$] (c) {};
 \draw (-1,1) node (a1) {};
 \draw (-1,-1) node (a2) {};
 \draw (3,1) node (c1) {};
 \draw (3,-1) node (c2) {};
 \draw (a)--(b)--(c)--(c2) (c)--(c1);
 \draw[dashed] (a1)--(a)--(a2);
\end{tikzpicture}} &
$a: X\rightarrow Y$ & $a:X\rightarrow Y$ & \multirow{3}{*}{$(2-\omega_1-\omega_2, 1-\omega_1-\omega_2+\omega_3)$}\\
&$b: Z\rightarrow Y$ & $b:Z\rightarrow U$ &\\
& & $c:U\rightarrow W$ &\\
& & &\\
\hline
\multicolumn{4}{l}{Case 4.(c)}\\
\hline
\multirow{4}{*}{\begin{tikzpicture}[scale=0.66]
 \tikzstyle{vertex}=[minimum size=2mm,circle,fill=black,inner sep=0mm]
 \tikzstyle{ptvertex}=[minimum size=2mm,rectangle,minimum width=4mm,fill=black,inner sep=0mm]
 \draw (0,0) node[vertex,label=below:$a$] (a) {};
 \draw (1,0) node[vertex,label=below:$b$] (b) {};
 \draw (2,0) node[vertex,label=below:$c$] (c) {};
 \draw (-1,1) node (a1) {};
 \draw (-1,-1) node (a2) {};
 \draw (3,1) node (c1) {};
 \draw (3,-1) node (c2) {};
 \draw (1,1) node[ptvertex] (b1) {};
 \draw (a2)--(a)--(b)--(c)--(c2) (c)--(c1);
 \draw[dashed] (a1)--(a) (b)--(b1);
\end{tikzpicture}} &
$a: U\rightarrow X$ & $a:U\rightarrow Y$ & \multirow{3}{*}{$(2\omega_1-\omega_3, 2)$}\\
&$b: W\rightarrow X$ & $b:W\rightarrow Y$ &\\
& & $c: U\rightarrow W$ &\\
& & &\\
\hline
\multicolumn{4}{l}{Case 4.(d)}\\
\hline
\multirow{4}{*}{\begin{tikzpicture}[scale=0.66]
 \tikzstyle{vertex}=[minimum size=2mm,circle,fill=black,inner sep=0mm]
 \tikzstyle{ptvertex}=[minimum size=2mm,rectangle,minimum width=4mm,fill=black,inner sep=0mm]
 \draw (0,0) node[vertex,label=below:$a$] (a) {};
 \draw (1,0) node[vertex,label=below:$b$] (b) {};
 \draw (2,0) node[vertex,label=below:$c$] (c) {};
 \draw (-1,1) node (a1) {};
 \draw (-1,-1) node (a2) {};
 \draw (3,1) node (c1) {};
 \draw (3,-1) node (c2) {};
 \draw (1,1) node (b1) {};
 \draw (a2)--(a)--(b)--(c)--(c1) (b)--(b1);
 \draw[dashed] (a1)--(a);
 \draw[black!30] (c)--(c2);
\end{tikzpicture}} &
$a: U\rightarrow X$ & $a:U\rightarrow Y$ & \multirow{3}{*}{$(\omega_1, 1+\omega_2)$}\\
& & $b:U\rightarrow Z$ &\\
& & &\\
& & &\\

\hline
\multicolumn{4}{l}{Case 5, $d_G(x)=d_G(c)=3$ and there is  $q \in (X\cap  (N(x) \cup N(c))$, w.l.o.g. $q\in N(c)$ }\\
\hline
\multirow{4}{*}{\begin{tikzpicture}[scale=0.66]
 \tikzstyle{vertex}=[minimum size=2mm,circle,fill=black,inner sep=0mm]
 \tikzstyle{ptvertex}=[minimum size=2mm,rectangle,minimum width=4mm,fill=black,inner sep=0mm]
 \draw (0,0) node[vertex,label=below:$a$] (a) {};
 \draw (1,0) node[vertex,label=below:$b$] (b) {};
 \draw (2,0) node[vertex,label=below:$c$] (c) {};
 \draw (-1,1) node (a1) {};
 \draw (-1,-1) node (a2) {};
 \draw (3,1) node (c1) {};
 \draw (3,-1) node (c2) {};
 \draw (1,1) node[vertex,label=right:$x$] (b1) {};
 \draw (a)--(b)--(c)--(c1) (b)--(b1);
 \draw[dashed] (a1)--(a)--(a2);
 \draw[black!30] (c)--(c2);
\end{tikzpicture}} &
$a: X\rightarrow Y$ &$a:X\rightarrow Y$ &\multirow{3}{*}{$(2-\omega_1,3-2\omega_1,1-\omega_1+\omega_2)$}\ \\
&$b: U\rightarrow Y$  & $b:U\rightarrow Z$ &\\
& (2nd branch) & &\\
& $q:X \rightarrow Y$&\\

\hline
\multicolumn{4}{l}{Case 5, $d_G(x)=d_G(c)=3$}\\
\hline
\multirow{4}{*}{\begin{tikzpicture}[scale=0.66]
 \tikzstyle{vertex}=[minimum size=2mm,circle,fill=black,inner sep=0mm]
 \tikzstyle{ptvertex}=[minimum size=2mm,rectangle,minimum width=4mm,fill=black,inner sep=0mm]
 \draw (0,0) node[vertex,label=below:$a$] (a) {};
 \draw (1,0) node[vertex,label=below:$b$] (b) {};
 \draw (2,0) node[vertex,label=below:$c$] (c) {};
 \draw (-1,1) node (a1) {};
 \draw (-1,-1) node (a2) {};
 \draw (3,1) node (c1) {};
 \draw (3,-1) node (c2) {};
 \draw (1,1) node[vertex,label=right:$x$] (b1) {};
 \draw (a)--(b)--(c)--(c1) (b)--(b1);
 \draw[dashed] (a1)--(a)--(a2);
 \draw[black!30] (c)--(c2);
\end{tikzpicture}} &
$a: X\rightarrow Y$ & $a:X\rightarrow Y$ & \multirow{3}{*}{$~(1-\omega_1+\omega_2, 2-\omega_1+\omega_2,2-\omega_1+\omega_2)~$}\\
&$b: U\rightarrow Y$ & $b:U\rightarrow Z$ &\\
&$c/x: U\rightarrow Z$ & &\\
& \multicolumn{3}{l}{There are $3$ branches; $2$ of them (add) are symmetric.}\\

\hline
\multicolumn{4}{l}{Case 5, $d_G(x)=2$ or $d_G(c)=2$ and }\\
\hline
\multirow{4}{*}{\begin{tikzpicture}[scale=0.66]
 \tikzstyle{vertex}=[minimum size=2mm,circle,fill=black,inner sep=0mm]
 \tikzstyle{ptvertex}=[minimum size=2mm,rectangle,minimum width=4mm,fill=black,inner sep=0mm]
 \draw (0,0) node[vertex,label=below:$a$] (a) {};
 \draw (1,0) node[vertex,label=below:$b$] (b) {};
 \draw (2,0) node[vertex,label=below:$c$] (c) {};
 \draw (-1,1) node (a1) {};
 \draw (-1,-1) node (a2) {};
 \draw (3,1) node (c1) {};
 \draw (3,-1) node (c2) {};
 \draw (1,1) node[vertex,label=right:$x$] (b1) {};
 \draw (a)--(b)--(c)--(c1) (b)--(b1);
 \draw[dashed] (a1)--(a)--(a2);
 \draw (c)--(c2);
\end{tikzpicture}} &
$a: X\rightarrow Y$ & $a:X\rightarrow Y$ & \multirow{3}{*}{$~(2-\omega_1, 2-\omega_1,2-\omega_1)~$}\\
&$b: U\rightarrow Y$ & $b:U\rightarrow Z$ &\\
& \multicolumn{3}{l}{When $\{a,b\}$ is deleted, \textbf{ConsDeg2} additionally decreases $k$ by $1$}\\
& \multicolumn{3}{l}{and removes a vertex of $Z$.}\\
\hline
\end{tabular*}
\caption{\label{tab:paramCA}Analysis of the branching for the running time of Theorem~\ref{thm:rtkii}}
\end{table}
 
The tight branching numbers are found for cases 4.(b) with $d_T(a)=2$, 4.(c), 4.(d), and 5. with all of $b$'s neighbors having degree $3$. The respective branching numbers are $(2-\omega_1-\omega_2,1-\omega_1-\omega_2+\omega_3)$, $(2\omega_1-\omega_3,2)$, $(\omega_1,1+\omega_2)$, and $(1-\omega_1+\omega_2,2-\omega_1+\omega_2,2-\omega_1+\omega_2)$. They all equal $\rtkii$.
\qed
\end{proof}
}

\section{Conclusion \& Future Research}
We have shown that {\sc Max Internal Spanning Tree} can be solved in time $\Oh^*(3^n)$.
In a preliminary version of this paper we asked if \mist can be solved in time $\Oh^*(2^n)$ and also expressed our interest in polynomial space algorithms for \mist. These questions have been settled very recently by Nederlof~\cite{Nederlof09} by providing a $\Oh^*(2^n)$ polynomial-space algorithm for \mist which is based on the principle of Inclusion-Exclusion and on a new concept called ``branching walks''.

This paper focuses on algorithms for \mist that work for the degree-bounded case, in particular,
for subcubic graphs. The main novelty is a Measure \& Conquer approach
to analyse our algorithm from a parameterized perspective (parameterizing by the solution size). We are not aware of many examples 
where this was successfully done without cashing the obtained gain at an early stage, see~\cite{Wah2007}. More examples in
this direction would be interesting to see. 
\longversion{Further improvements on the running times of our algorithms pose another natural challenge.}

A closely related problem worth investigating is the generalisation to directed graphs: Find a directed tree, which consist of directed paths form the root to the leaves with as few leaves as possible. Which results can be carried over to the directed case?\\[1ex]
\textbf{Acknowledgment} We would like to thank Alexey A. Stepanov for useful discussions in the initial phase of this paper.


\newpage

\appendix

\end{document}